\documentclass[11pt]{article}
\usepackage[margin=3cm]{geometry}
\usepackage{amsmath, amsthm}
\usepackage{natbib}
\usepackage{graphicx}
\usepackage{epstopdf}
\usepackage{epsfig}

\theoremstyle{plain}
\newtheorem{theorem}{Theorem}[section]

\newtheorem{lemma}[theorem]{Lemma}

\theoremstyle{definition}

\theoremstyle{remark}

\newcommand{\tr}{^{\prime}}
\def\b#1{\mbox{\boldmath $#1$}}    

\newcommand{\E}{{\rm E}}
\newcommand{\diag}{{\rm diag}}    

\title{A Fisher-scoring algorithm for fitting latent class models with individual covariates}
\author{A. Forcina,
Dipartimento di Economia, Finanza e Statistica,\\University of
Perugia, Italy}
\begin{document}
\maketitle
\begin{abstract}
We describe a modified Fisher scoring algorithm for fitting a wide variety of latent class models where the class weights and the conditional distributions of the responses may depend on continuous covariates through a multinomial logit model. We derive a simple expression for the score vector, which can be computed efficiently, together with the empirical information matrix, which we show to be an unbiased estimate of both the observed and the expected information matrices. The proposed Fisher scoring algorithm uses the empirical information matrix to update the step direction while the step length is determined by a line search routine which ensures that the likelihood never decreases from one step to the next. The new algorithm converges rather quickly for any choice of starting values as long as they are far away from the boundary of the parameter space. A convenient expression for the observed information matrix is also derived. An application to education transmission is used as an illustration.
\end{abstract}


\section{Introduction}
Latent class models are important tools for modeling the dependence structures induced by unobservable heterogeneity with important applications in several fields of social and biomedical sciences. Since the pioneering work of \cite{Goodman}, extensions in several directions have been proposed. Particular attention has been devoted to the case where the probabilities of belonging to latent classes may depend on covariates, see, for a detailed discussion, \cite{Verm2010} and \cite{Band:2012}. Latent class models where covariates may also affect the conditional distribution of the responses conditionally on the latent have been used, for instance, by \cite{BaFo:2006} and \cite{DaLi:2012} among others. For a general discussion of identifiability and inference in latent class models where covariates may affect both the marginal distribution of the latent and the conditional distribution of the responses, see \cite{HuaBan:04}. Finally, models where the assumption of conditional independence may be violated have been considered, among others, by \cite{Hagenaars} within a log-linear context and \cite{BaFo:2006} in the context of marginal models.

The EM (expectation-maximization) algorithm is generally used to compute the maximum likelihood estimates, though, for instance, the Latent GOLD software combines EM and Newton-Raphson. The EM algorithm is known for being numerically stabile and the likelihood always increases from one step to the next. The Newton-Raphson algorithm, on the other hand, is likely to diverge, unless the starting values are very close to a local maximum. The performance of quasi-Newton algorithms may be greatly improved by performing a suitable line search to optimize the step length \cite[see for example][]{PotShi:95, Turn:2008}; the algorithm that we propose adopts simple
diagnostics to determine the behaviour of the likelihood locally and to optimize the step length.

One drawback of the EM algorithm is that convergence may be slow; this is a  problem when several covariates of interest are available and one needs to search for a suitable model among many different alternatives. In the case when covariates are assumed to affect only the latent weights, simple estimates may be computed by the following three step procedure: (i( fits an ordinary latent class model without covariates, (ii) on this basis,  assigns subjects to latent classes and (iii) models the effect  of covariates as if the latent class was known. \cite{BoCrHa2004} have shown that this lead to biased estimates and proposed a correction to the third step. Recently, \cite{Verm2010} has suggested a more general formulation of the third step which provides better correction.

In this paper we propose a Fisher scoring kind of algorithm for fitting a more general family of latent class models where individual covariates may affect both the marginal distribution of the latent and the conditional distributions of the responses. In addition, our framework allows for the response variables to be associated conditionally on the latent, as long as the model must be identifiable. By introducing a convenient matrix notation, we derive a simple expression for the score vector and the {\it Empirical} Fisher information matrix, \cite{Meil:89}, which we suggest using for computing the step direction. This matrix, which we show to be positive definite under very mild conditions, is much simpler to compute than both the observed and the expected information matrices. The expected information matrix is always positive definite whenever the model is identifiable; on the other hand, it has been argued \cite[see][]{EfrHin:1978} that the observed information matrix is preferable for the asymptotic distribution of the maximum likelihood estimates. The empirical information matrix is, somewhere, in between the two and makes the algorithm much more efficient. The use of the empirical information matrix in fitting algorithms and for estimating standard errors has been investigated by \cite{Scott:02} in a quite general context.

In section 2, after defining the model, we derive an expression for the score and the empirical information matrix and give conditions under which this matrix is positive definite. In section 3 we discuss efficient computation of these quantities and describe a line search algorithm which makes the proposed Fisher scoring algorithm efficient and very stable at the same time. An expression for computing the observed information matrix is also given in the Appendix. In section 4 we present an application from the field of education transmission.
\section{Description of the model and main results}
In the following we write $\b 1_c$ and $\b I_r$ to denote, respectively, a column vector of $c$ ones and an identity matrix of size $r$. Suppose there are $n$ subjects and $c$ latent classes with the latent variable $U$  coded as $0, 1, \cdots, c-1$. Let $\b\pi_i$, $i=1,\dots, n$, be the vector whose elements are the probabilities that the $i$th subject belongs to one of the $c$ latent classes; let $\b X_i$ be a $c\times k$ matrix of known constants, possibly depending on individual covariates, and assume that the marginal distribution of the latent is determined by a multinomial logit model with the initial category as reference,
$$
\b\pi_i=\frac{\exp(\b X_i\b\beta)}{\b 1_c\tr \exp(\b X_i\b\beta)}.
$$
When no covariates are available, $\b X_i$ will be an identity matrix of size $c$ without the first column. To make the $j$th logit ($j=1,\dots ,c-1$) depend on $x_{ik}$, simply add a column of zeros except for the $j+1$th element which is equal to $x_{ik}$.

Let $r$ be the number of possible configurations of the response variables; for the $i$th subject, their joint distribution conditional on the latent $U=j$, $j=0,\dots, c-1$, is determined by the $r\times 1$ vector of probabilities $\b q_{ij}$ whose entries are arranged in lexicographic order, meaning that response variables with a larger index run faster. We assume a multivariate logit model
$$
\b q_{ij}=\frac{\exp(\b G\b\theta_{ij})}{\b 1_r\tr \exp(\b G\b\theta_{ij})},
$$
where $\b G$ is a $r\times g$ full rank design matrix and $\b\theta_{ij}$ is a vector of log-linear parameters. If responses are conditionally independent, $\b G$ will code their main effects which correspond to the marginal logits relative to the initial category. When, like in the application presented below,  certain responses are assumed to depend on others, suitable additional columns will have to be added, like in a design matrix coding main effects and interactions.

When covariates are assumed to affect the conditional distribution of the responses, a linear model on the logits is assumed
$$
\b\theta_{ij} = \b Z_i\b\gamma_j;
$$
this formulation implies that regression parameters are specific to each latent class, but the structure of the model is the same. More parsimonious models may be formulated by imposing linear restrictions on the logits; for instance we might assume that the effect of certain covariates does not depend on the latent. For a general formulation it is convenient to stack the vectors $\b\gamma_j$ one below the other into the vector $\b\gamma$, then any restriction on the linear model above may be expressed as
$$
\b B\b\gamma=\b 0,
$$
where $\b B$ is a suitable matrix of row contrasts. For instance, to constraint two elements of $\b\gamma$ to be equal, $\b B$ will contain a row of 0's except for -1 and 1 in the entries corresponding to the two parameters assumed to be equal.

The formulation we are proposing here does not assume  that the responses are independent conditionally on the latent; the conditional dependence structure is determined by the design matrix $\b G$ which can take any form, as long as the model is identifiable. Finally, let $\b Q_i$ be the matrix whose $j$the column is $\b q_{ij}$, so that $\b p_i$ = $\b Q_i\b\pi_i$ is the marginal distribution of the responses for the $i$th subject.

Let $\b y_i$ denote a $r\times 1$ vector of  0's except for the entry corresponding to the observed response configuration for subject $i$. Under the assumption that observations from different subjects are independent, the log-likelihood may be written as $\ell(\b\beta,\b\gamma)$ = $\sum_I \b y_i\tr\log(\b p_i)$.
\subsection{The score vector}
By applying the chain rule to differentiate the log-likelihood, the score relative to $\b\beta$ may be written as
$$
\b s_{\b \beta} = \sum_i\frac{\partial \b\pi_i\tr}{\partial \b\beta} \frac{\partial \b p_i\tr}{\partial \b\pi_i} \frac{\partial \ell(\b\beta,\b\gamma)}{\partial \b p_i} = \sum_i \b X_i\tr\b\Omega_{\b \pi_i}\b Q_i\tr\diag(\b p_i)^{-1}\b y_i,
$$
where $\b\Omega_{\b \pi_i}$ = $\diag(\b\pi_i)-\b\pi_i\b\pi_i\tr$.
Noting that $\b p_i$ = $\sum_j\pi_{ij}\b q_{ij}$, the score relative to $\b\gamma_j$ is
$$
\b s_j = \sum_i \frac{\partial \b q_{ij}\tr}{\partial \b\gamma_j} \frac{\partial \b p_i\tr}{\partial \b q_{ij}} \frac{\partial \ell(\b\beta,\b\gamma)}{\partial \b p_i} = \sum_i \pi_{ij}\b Z_i\tr\b G\tr\b\Omega_{ij}\diag(\b p_i)^{-1} \b y_i,
$$
where $\b\Omega_{ij}$ = $\diag(\b q_{ij})-\b q_{ij}\b q_{ij}\tr$.
Let $\b\psi$ be the vector obtained by stacking $\b\beta$ and $\b\gamma$ one below the other and define
$$
\b D_i = \frac{\partial \b p_i}{\partial\b\psi\tr} = \begin{pmatrix}
\b Q_i\b\Omega_{\b \pi_i}\b X_i & \pi_{i1}\b\Omega_{i1}\b G\b Z_i & \dots & \pi_{ic}\b\Omega_{ic}\b G\b Z_i\end{pmatrix};
$$
then the whole score vector may be written as
\begin{equation}
\b s = \frac{\partial \ell(\b\psi)}{\partial \b\psi} =\sum_i \b D_i\tr \diag(\b p_i)^{-1}\b y_i.
\label{eq:score}
\end{equation}

When constraints across conditional distributions are present, recall that $\b B\b\gamma$ = $\b 0$ implies a linear model of the form $\b\gamma$ = $\b R\b\tau$, where $\b R$ spans the orthogonal complement to the space spanned by the columns of $\b B\tr$, that is $\b B\b R=\b 0$. In words, the constrained model may be fitted by replacing $\b\gamma$ with the vector of free parameters $\b\tau$. In practice, we simply need to left multiply the score vector by the block diagonal matrix with elements $(\b I_k,\:\b R\tr)$. In the most common cases where two or more elements of $\b\gamma$ are constrained to be equal, the matrix $\b R$ may be constructed explicitly as follow: (i) start with $\b R$ = $\b I_{ch}$, where $h$ is the dimension of $\b\gamma_j$, $j=0,\dots, c-1$, (ii) process constraints one at a time; (iii) if, say, the entries $l_1$ and $l_2$ of $\b\gamma$ must be constrained to be equal, in the actual version of $\b R$ drop column $l_2$ and replace column $l_1$ with the sum of columns $l_1$ and $l_2$.
\subsection{The information matrices}
In order to derive an expression for the observed information matrix, let $\b A,\:\b b, \:\b x$ be respectively $a\times b$, $b\times 1$ and $t\times 1$ and suppose that $\b b$ is a vector of constants, define
$$
\b\Delta(\b A,\b b,\b u) = \begin{pmatrix}\frac{\partial \b A}{\partial u_1}\b b & \dots & \frac{\partial \b A}{\partial u_t}\b b \end{pmatrix};
$$
the following result may be verified by direct expansion
\begin{equation}
\b\Delta(\b A,\b b,\b u) = \b\Delta(\b A,\b b,\b v)\frac{\partial\b v}{\partial\b u\tr}.
\label{Delta}
\end{equation}
Because $\E(\b y_i)$ = $\b p_i$ and $\b\Omega_{\b \pi_i}$, $\b\Omega_{ij}$ are both orthogonal to a unitary vector of appropriate size, it follows easily that
\begin{equation}
\E[\b D_i\tr\diag(\b p_i)^{-1}\b y_i]=\b 0
\label{rank0}
\end{equation}
Let $\b F_{ex},\:\b F_{em},\:\b F_{os}$ denote respectively the expected, empirical and observed information matrix. Let $l_i$ denote the element of $\b y_i$ which equals 1 and let $\b S$ be the $(k+ch)\times n$ matrix whose $i$th columns is the $l_i$th column of $\b D_i\tr/p_{ih_i}$; from (\ref{eq:score}), it follows that the $\b s$ may be recovered by stacking the columns of $\b S$ one below the other. Note also that
$$
\b F_{em} = \b S\b S\tr.
$$

\begin{lemma}
$\b F_{os}-\b F_{em}$ = $-\sum_i\b\Delta(\b D_i\tr,\diag(\b p_i)^{-1}\b y_i,\b\psi)$, in addition, $\E(\b F_{os}-\b F_{em})$ = $\E(\b F_{ex}-\b F_{em})$ = $\b 0$.
\end{lemma}
\begin{proof}
See the Appendix where we derive also an expression for $\b F_{os}$.
\end{proof}

\begin{lemma}
The empirical information matrix is positive definite if and only if $\b S$ is of full row rank.
\end{lemma}
\begin{proof}
Follows from the fact that $\b F_{em}$ = $\b S\b S\tr$ is symmetric; a necessary condition for $\b S$ to be of full rank is that $n\geq k+ch$, implying that the number of observations must not be  smaller than the number of free parameters.
\end{proof}
\section{Description of the algorithm}
Here we propose to compute the maximum likelihood estimates by a modified Fisher scoring algorithm with $\b F_{ex}$ replaced by $\b F_{em}$, combined with a suitable line search whose features are outlined below. $\b F_{em}$ seems to be a suitable compromise between the expected and the observed information matrices; the latter in addition to being much harder to compute, it is also likely to be non positive definite unless the estimate is close to a local maximum. As we show below, $\b F_{em}$ is much easier to compute than $\b F_{ex}$. A comparison between the estimates of the standard errors provided by $\b F_{em}$, $\b F_{ex}$, and $\b F_{os}$ is presented in Section 4. A set of {\sc MatLab} functions that implements the algorithm described here will be provided as supplementary material.
\subsection{Computational aspects}
\label{Fstep}
Now we show that, having assumed that individual observations are available, computation of the score vector and of the empirical information matrix can be performed in an efficient way. The hardest step is computation of the $\b D_i$ matrices which involve, apparently, several matrix multiplications. We now show that for this task, the only matrix multiplication that is required is $\b G\b Z_i$; however, because these matrices do not depend on parameters, the product needs to be computed only once. Let $l_i$ denote the entry of $\b y_i$ which is equal to 1, let also $\tilde{\b q}_i$ and $\tilde{\b g}_j$ denote, respectively, the $l_i$th column of $\b Q_i\tr$ and of $(\b G\b Z_i)\tr$, then, the $i$th columns of $\b S$ may be computed as
$$
\begin{pmatrix}
\b X_i\tr[\diag(\b\pi_i)\tilde{\b q}_{li}-\b\pi\b\pi\tr\tilde{\b q}_{li}]\\
\pi_{i0} [\diag(\b q_{i0})\tilde{\b g}_i-\b q_{i0}\b q_{i0}\tr \tilde{g}_i]\\
\cdots  \\
\:\pi_{i,c-1} [\diag(\b q_{i,c-1})\tilde{\b g}_i-\b q_{i,c-1} \b q_{i,c-1}\tr \tilde{g}_i]\:
\end{pmatrix}.
$$
\subsection{Line search}
After $s-1$ steps, the basic updating equation takes the form
$$
\hat{\b\psi}^{(s)} = \hat{\b\psi}^{(s-1)}+a_{s-1} \left( F_{em}^{(s-1)}\right)^{-1}\b s^{(s-1)},
$$
where $a_{s-1}$ is the step length. When the log-likelihood is not concave, an algorithm with $a_s=1$ is almost certain to diverge, unless the starting value is very close to a local maximum; one possibility would be to set $a_0$ very small and let it increase slowly with $s$. In a related context, \cite{Turn:2008} suggest using the Levenberg--Marquardt algorithm which combines Newton-Raphson and steepest ascent steps; this would be less efficient in our context where the information matrix is positive definite. Our algorithm uses a proper line search where the log-likelihood is never allowed to decrease. Its main features are given below; before starting, set $a_0$ to some value possibly smaller than 1, say, 0.5;
\begin{enumerate}
\item
compute the log-likelihood, the score and the information matrix at $\hat{\b\psi}^{(h)}$, the estimate computed in the previous step, and use the updating equation with the last available step length to compute the initial $\hat{\b\psi}^{(h,a)}$ for the next step;
\item
compute the log-likelihood and the score at $\hat{\b\psi}^{(h,a)}$, find the step length that maximizes a cubic approximation to the log-likelihood along the given direction and compute the second guess $\hat{\b\psi}^{(h,b)}$;
\item
update the estimate by choosing the guess value with the highest likelihood;
\item
in case of no improvement, first shorten the step and, if even this does not work, perform a steepest ascent step.
\end{enumerate}
\subsection{Discussion}
The algorithm has been tested in several different contexts and it has always performed very well. Typically, when using a starting point selected at random and possibly not too far from 0, the algorithm takes small steps in the initial stage, but then goes very quickly to a local maximum. However, if one or more parameter estimates tend to large absolute values, the algorithm may become slower; in such cases one might shift such estimates towards the origin and restart the algorithm in the hope to find a local maximum away from the boundary of the parameter space. In any case, to increase the probability of reaching a global maximum, after convergence, a random perturbation may be applied to the estimates and the algorithm restarted a few times.

Compared with a standard general purpose EM algorithm, our new algorithm has shown to be 20 to 50 times faster. An explanation for this seems to be that, given the nature of the assumed model, each M-step of the EM algorithm is not much simpler than a step of our Fisher scoring algorithm. Actually, the latter is usually considerably faster with individual data because the response vector $\b y_i$ reduces to the index of the response configuration for the $i$th subject. Similar simplifications cannot be applied to the {\it complete data} for the same subject because the vector produced in the E-step has $c$ entries different from 0. The following experiment was also preformed in the context of the data presented in section 4: the Fisher scoring algorithm applied to the original data with 2568 subjects was compared with an EM algorithm applied to the same data aggregated into 125 non empty response configurations determined by discretizing each covariate at its four quantiles. The Fisher scoring was considerably faster even in this context.
\section{Application}
\subsection{The data}
As an illustration, we analyse a data set from the National Child Development Survey (NCDS), a UK cohort study that included everybody born in UK, March 1958, from the 3rd to the 9th. Information on family background, schooling and social achievements for the subjects in the sample were collected at different stages of their lives. In the application below we consider the joint distribution of the following response variables: performance in mathematics and reading test scores taken when the child was 10 and 16 years old, $M_{10},\:M_{16},\:R_{10},\:R_{16}$, an overall measure of non cognitive attitudes (as reported by teachers), $N$ and the academic qualification achieved $Q$ (none, O-level, A-level, university degree). All variables except the last were coded into three categories based on quantiles. As covariates, we consider the amount of education (measured by the number of years in school) for each parent, $E_f,\:E_m$; in addition we computed an overall measure of the concern $C_p$ for the child education on the side of the parents based on reports from the teachers.
For simplicity, in this application we restrict attention only to the 2568 females with no missing data for the selected variables. Because the relative size of the selected sub-sample is slightly less than $30\%$, results should be interpreted with care.
A complete description of the original data is available at \texttt{http://www.esds.ac.uk/longitudinal/access/ncds}.
\subsection{Model selection}
The issue underlying this analysis is wether the observed correlation between test scores and academic qualification can be mostly explained by latent abilities, though family pressure may also play a role. We start by assuming that $E_f,\:E_m$ (parents' education) affect the marginal distribution of the latent as in the equation below
$$
\log\frac{\pi_{ij}}{\pi_{i0}}=\beta_{j,1}+\beta_{j,2} E_{f,i} + \beta_{j,3}E_{m,i}, \: j=1, \dots c-1;
$$
in addition to the $c-1$ logit intercepts this model requires $2(c-1)$ regression coefficients. The response variables are assumed conditionally independent, except for a first order autoregressive model within the two pairs $(M_{10},\:M_{16})$ and $(R_{10},\:R_{16})$. To avoid estimating the two $3\times3$ transition matrices for Math and Read, we fitted a parsimonious model based on scores, for instance, for Math we assume that,
$$
\log\frac{P(M_{16}=t\mid U=j,M_{10}=s)}{P(M_{16}=0\mid U=j,M_{10}=s)} =
\theta_{j1}^M+\theta_{j2}^M[\delta(s,t)+0.5 \delta(\mid s-t\mid ,1)],
$$
where $j=0,\dots ,c-1$, $t=1,2$, $s=0,1,2$ and $\delta(a,b)=0$ unless $a=b$; a similar equations was used for the read scores. To account for the fact that the academic qualification achieved by the child may also be affected by pressure on the family side, we allowed the logits of $Q$ to depend on $C_p$, the concern on the side of the parents:
$$
\log\frac{P(Q=h\mid U=j,C_{pi})}{P(Q=0\mid U=j,C_{pi})} =
\theta_{ijh}^Q = \gamma_{jh1}+\gamma_{jh2} C_{pi}.
$$

Initially, the basic model described above with an increasing number of latent classes was fitted until the BIC started to increase: the best among such models had three latent classes as can be seen from Table \ref{tab:1}. Inspection of the 9 regression coefficients, three for each latent class,  measuring the effect of $C_p$ on each logit of $Q$, suggested to constraint to 0 four of these which were non significant at the 5\% level, this is model LC3a. In principle, parents' concern might, instead, affect directly the marginal distribution of the latent: children who receive more attention tend to belong to more talented latent types; this is model LC3b. Clearly model LC3a is the one which performs best among those considered here; the results below refer to this specific model.
\begin{table}
\caption{Bayes information criteria}
\centering
\begin{tabular}{lcrrrrr} \\ \hline
Model & \hspace{1mm} & LC2 & LC3 & Lc4 & LC3a & LC3b
\\ \hline
d.f. & & 39 & 60 & 81 & 56 & 53 \\
BIC & & 28022 & 27615 & 27677 & 27606 & 27644 \\
 \hline
\end{tabular}
\label{tab:1}
\end{table}
\subsection{Main results}
With family concern fixed at its average, the cumulative conditional distributions of academic qualification, Math and Read score at 16  clearly indicate that, after a suitable re-labeling of the latent types, these may be ordered from least to most talented/successful as shown in Table \ref{tab:2}. Subjects in latent type 0 perform very badly in Math and Read and, with high probability, do not reach any qualification while subject in latent type 2 have almost 1 chance out of 3 to get a university degree.
\begin{table}
\caption{Conditional cumulative  distributions of academic qualification, Math and Read score at 16}
\centering
\begin{tabular}{lrrrcrrcrr} \\ \hline
 & \multicolumn{3}{c}{$Q$} &\hspace{0.5mm}& \multicolumn{2}{c}{$M_{16}$} &\hspace{0.5mm}& \multicolumn{2}{c}{$R_{16}$}
 \\ \hline
$U$ & None & O-lev & A-lev & & 0 & 1 & & 0 & 1 \\ \hline
0 & 0.8903 & 0.9665 & 0.9705 & & 0.8082 & 0.9974 && 0.8301 & 1.0000 \\
1 & 0.3606 & 0.8637 & 0.8944 & & 0.3162 & 0.8931 && 0.2114 & 0.8073 \\
2 & 0.0773 & 0.4644 & 0.6719 & & 0.0114 & 0.1977 && 0.0106 & 0.2710
 \\ \hline
\end{tabular}
\label{tab:2}
\end{table}

All the regression coefficients displayed in Table \ref{tab:3} are positive and significant. As concerns the effect of parents' education, positive coefficients imply that, as the education of either parent increase, the probability of belonging to latent types 1 or 2 rather than 0 also increases. Note also that the education of the mother seems to be more effective than that of the father and that a higher education of either parent is stronger in moving from type 0 to type 2.
\begin{table}
\caption{Estimated regression  coefficients and standard errors}
\centering
\begin{tabular}{lrrlrr} \\ \hline
 & \multicolumn{2}{c}{$U=1/U=0$} & &\multicolumn{2}{c}{$U=2/U=0$}  \\
 & $\hat\beta$ & $se$ &  & $\hat\beta$ & $se$
 \\ \hline
$E_f$ & 0.1272 & 0.0785 & $E_f$ & 0.1939 & 0.0693 \\
$E_m$ & 0.2613 & 0.0662 & $E_m$ & 0.4849 & 0.0621
\\ \hline
 & \multicolumn{5}{c}{Parents' concern}
\\ \hline
 & $\hat\gamma$ & $se$ & & $\hat\gamma$ & $se$
\\ \hline
$\frac{Q=1}{Q=0}|U=0$ & 0.0949 & 0.0399 & $\frac{Q=3}{Q=0}|U=1$ & 0.1970 & 0.0465 \\
$\frac{Q=2}{Q=0}|U=0$ & 0.3187 & 0.1352 & $\frac{Q=3}{Q=0}|U=2$ & 0.4560 & 0.1330 \\
$\frac{Q=3}{Q=0}|U=0$ & 0.1534 & 0.0703 & & &
 \\ \hline
\end{tabular}
\label{tab:3}
\end{table}

Finally, the regression coefficients for $M_{10},\:R_{10}$ in the equation for $M_{16},\:R_{16}$ respectively, are all positive and highly significant; those in Read are also greater. This seems to suggest  that the correlation between test scores of the same kind taken at the age of 10 and 16 are strongly correlated even conditionally to the latent type.
\begin{table}
\caption{Regression coefficient for test scores at 16 relative to test score at 10}
\centering
\begin{tabular}{lrrrrrr} \hline
 & \multicolumn{2}{c}{U=0} & \multicolumn{2}{c}{U=1} & \multicolumn{2}{c}{U=2} \\
 & $\hat\theta$ & $se$ & $\hat\theta$ & $se$ & $\hat\theta$ & $se$
 \\ \hline
Math & 0.6971 & 0.2783 & 0.7103 & 0.1746 & 1.5721 & 0.2485 \\
Read & 1.4922 & 0.2847 & 1.3260 & 0.1599 & 1.6889 & 0.1964
 \\ \hline
\end{tabular}
\label{tab:4}
\end{table}

In Figure 1 below we have plotted the estimates of the standard errors computed from the observed and the empirical information matrices for the final model: if we consider the estimates provided by the observed information matrix to be the most reliable, those provided by the empirical information matrix seem to provide a very close approximation. The plot does not includes the 5 intercept parameters which are greater than 20 in absolute value because the corresponding standard errors are huge: this simply means that, when certain estimated probabilities are very close to 0 or 1, the data provide little information about how close to the boundary these estimates are going to be.
\begin{center}
\begin{figure}
\includegraphics[width=12.8cm,height=6.0cm]{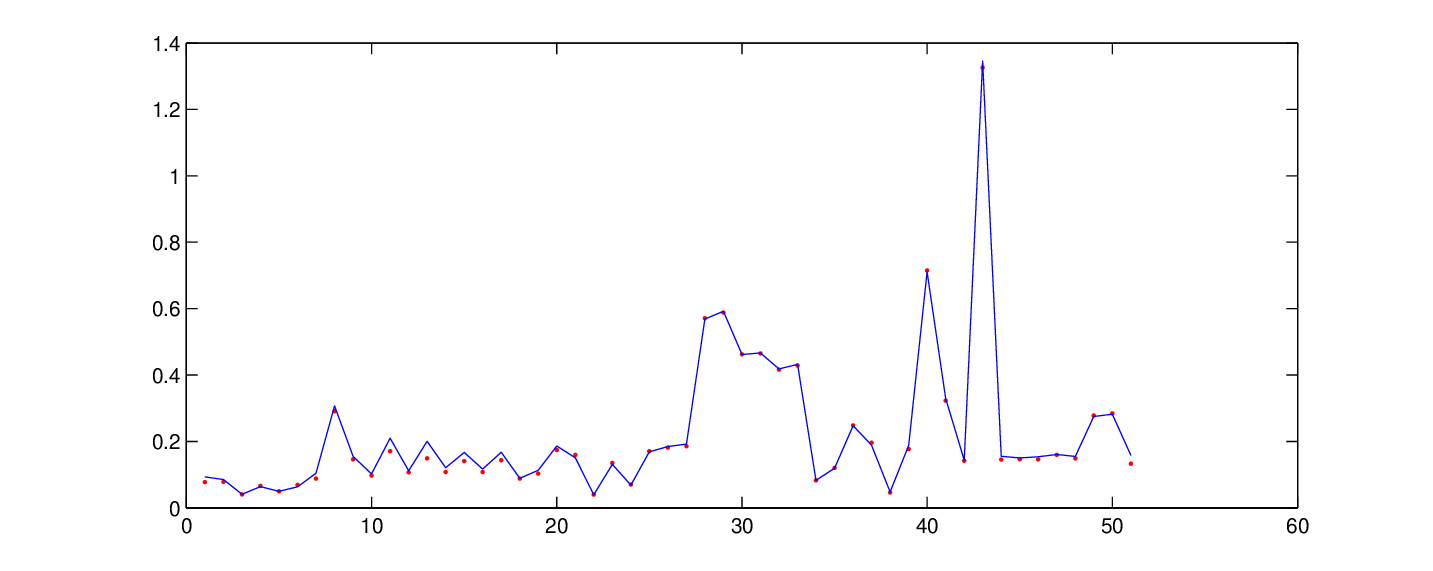}
\caption{Estimates of the standard errors from the observed (solid line) and empirical (dots) information matrix.}
\end{figure}
\end{center}
\subsection*{Acknowledgements}
I would like to thank J.K. Vermunt (Tilburg University) and M. Alfo' (Rome Universitu) for comments and suggestions.
\section*{Supplementary material}
A set of {\sc MatLab} functions that implement the algorithm described here are available from the Author on request. One of these functions will be a sample script that loads the data used in the Application and invokes the main fitting algorithm.
%
\bibliographystyle{spbasic}
\bibliography{bib}

\begin{thebibliography}{13}
\providecommand{\natexlab}[1]{#1}
\providecommand{\url}[1]{{#1}}
\providecommand{\urlprefix}{URL }
\expandafter\ifx\csname urlstyle\endcsname\relax
  \providecommand{\doi}[1]{DOI~\discretionary{}{}{}#1}\else
  \providecommand{\doi}{DOI~\discretionary{}{}{}\begingroup
  \urlstyle{rm}\Url}\fi
\providecommand{\eprint}[2][]{\url{#2}}

\bibitem[{Bartolucci and Forcina(2006)}]{BaFo:2006}
Bartolucci F, Forcina A (2006) A class of latent marginal models for
  capture-recapture data with continuous covariates. Journal of the American
  Statistical Association 101:786--794

\bibitem[{Bolck et~al(2004)Bolck, Croon, and Hagenaars}]{BoCrHa2004}
Bolck A, Croon M, Hagenaars J (2004) Estimating latent structure models with
  categorical variables: One-step versus three-step estimators. Political
  Analysis 12:3--27

\bibitem[{Dardanoni and Li~Donni(2012)}]{DaLi:2012}
Dardanoni V, Li~Donni P (2012) Incentive and selection effects of medicap
  insurance on impatient care. Journal of Health Economics 33:457--470

\bibitem[{Efron and Hinkley(1978)}]{EfrHin:1978}
Efron B, Hinkley D (1978) Assessing the accuracy of maximum likelihood
  estimates: observed versus expected information matrix. Biometrika
  65:457--487

\bibitem[{Goodman(1974)}]{Goodman}
Goodman L (1974) Exploratory latent structure analysis using both identifiable
  and unidentifiable models. Biometrika 61:215--231

\bibitem[{Hagenaars(1988)}]{Hagenaars}
Hagenaars J (1988) Latent structure models with direct effects between
  indicators: local dependence models. Sociological methods and research
  16:379--405

\bibitem[{Huang and Bandeen-Roche(2004)}]{HuaBan:04}
Huang G, Bandeen-Roche K (2004) Building an identifiable latent class model,
  with covariate effects on underlying and measured variables. Psychometrika
  69:5..32

\bibitem[{Meilison(1989)}]{Meil:89}
Meilison I (1989) A fast improvement to the em algorithm on its own terms.
  Journal of the Royal Statistical Society B 51:127--138

\bibitem[{Petersen et~al(2012)Petersen, Bandeen-ROche, Butz.Jorgensen, and
  Larsen}]{Band:2012}
Petersen J, Bandeen-ROche K, ButzJorgensen E, Larsen K (2012) Predicting latent
  class scores for subsequent analysis. Psychometrika 77(2):244--262

\bibitem[{Potra and Shi(1995)}]{PotShi:95}
Potra F, Shi Y (1995) An efficient line search algorithm for unconstrained
  optimization. Journal of Optimization Theory and Applications 85:677--704

\bibitem[{Scott(2002)}]{Scott:02}
Scott W (2002) Maximum likelihood estimation using the empirical fisher
  information matrix. Journal of Statistical Computation and Simulation 72
  (8):599--611

\bibitem[{Turner(2008)}]{Turn:2008}
Turner R (2008) Direct maximization of the likelihood of a hidden markov model.
  Computational Statistics and Data Analysis 52:4147--4160

\bibitem[{Vermunt(2010)}]{Verm2010}
Vermunt J (2010) Latent class modeling with covariates: Two improved three-step
  approaches. Political Analysis 18:450--469

\end{thebibliography}
\section*{Appendix}
The auxiliary result below may be verified by direct calculations:
\begin{equation}
\b\Delta[\diag(\b x)-\b x\b x\tr, \b b,\b x]=\diag(\b b)-\b I(\b x\tr\b b)-\b x\b b\tr.
\label{omega}
\end{equation}
The results below are derived by using the $\b\Delta$ notation and equations (\ref{Delta}), (\ref{omega}).
Let $\b v_i$ = $\b Q_i\tr \diag(\b p_i)^{-1}\b y_i$, the $(\b\beta,\b\beta\tr)$ diagonal block of $\b F_{os}-\b F_{em}$ may be computed as
$$
\sum_i\b\Delta(\b X_i\tr \b\Omega_{\b \pi_i},\b v_i,\b\beta) = \sum_i\b X_i\tr[\diag(\b v_i)-(\b\pi_i\tr\b v_i)\b I-\b\pi_i\b v_i\tr]\b\Omega_{\b \pi_i}\b X_i.
$$
Now let $\b v_i$ = $\diag(\b p_i)^{-1}\b y_i$, the diagonal block corresponding to $(\b\gamma_j,\b\gamma_j\tr)$ has the form
$$
\sum_i\b\Delta(\b\pi_{ij}\b Z_i\tr\b G\tr\b\Omega_{ij},\b v_i,\b\gamma_j) = \sum_i \b\pi_{ij}\b Z_i\tr\b G\tr[\diag(\b v_i)-(\b q_{ij}\tr\b v_i)\b I-\b q_{ij}\b v_i\tr]\b\Omega_{ij}\b G\b Z_i.
$$
Let again $\b v_i$ = $\diag(\b p_i)^{-1}\b y_i$, then the off-diagonal blocks indexed by $(\b\beta,\b\gamma_j\tr)$ may be computed as
\begin{eqnarray*}
\sum_i\b\Delta(\b X_i\tr \b\Omega_{\b \pi_i}\b Q_i\tr,\b v_i,\b\gamma_j) &=& \sum_i \b X_i\tr \b\Omega_{\b \pi_i} \b\Delta(\b Q_i\tr,\b v_i,\b q_{ij})\b\Omega_{ij}\b G\b Z_i \\ &=& \sum_i\b X_i\tr \b\Omega_{\b \pi_i} \b e_j\b v_i\tr \b\Omega_{ij}\b G\b Z_i,
\end{eqnarray*}
where $\b e_j$ is a vector of 0's except for the $j$th element which is 1. Finally the entry corresponding to all blocks indexed by $(\b\gamma_j,\b\gamma_h\tr)$ with $j\neq h$ are matrices of 0's.

\end{document}